\documentclass[sn-mathphys]{sn-jnl}

\jyear{2021}%

\theoremstyle{thmstyleone}%
\newtheorem{theorem}{Theorem}
\newtheorem{proposition}{Proposition}%
\usepackage[euler]{textgreek}
\usepackage{amscd,amssymb,amsmath,amsthm}

\usepackage{xcolor}

\usepackage{epstopdf}
\theoremstyle{thmstyletwo}%
\newtheorem{remark}{Remark}%
\newtheorem{lemma}{Lemma}
\theoremstyle{thmstylethree}%

\begin{document}

\title[Gibbs measure for the HC-Blume-Capel model ...]{Gibbs measure for the HC-Blume-Capel model in the case of a "wand" type graph on a Cayley tree}


\author*[1,2]{\fnm{Nosirjon M.} \sur{Khatamov}}\email{nxatamov@mail.ru}

\author[3]{\fnm{Malika A.} \sur{Kodirova}}\email{malika.kodirova24@gmail.com}
\equalcont{These authors contributed equally to this work.}

\affil*[1]{\orgdiv{Mathematics}, \orgname{Namangan State University}, \orgaddress{\street{Davlatabad district, Tadbirkor MFY, Yangi shahar street, 30}, \city{Namangan}, \country{Uzbekistan}}}

\affil[2]{\orgdiv{Mathematics}, \orgname{Institute of Mathematics named after V.I. Romanovsky of the Academy of Sciences of the Republic of Uzbekistan}, \orgaddress{\street{University street}, \postcode{100174}, \city{Tashkent}, \country{Uzbekistan}}}

\affil[3]{\orgdiv{Mathematics}, \orgname{Andijan State University}, \orgaddress{\street{Universitetskaya street}, \city{Andijan}, \postcode{170100}, \country{Uzbekistan}}}


\abstract{In this paper, we investigate translation-invariant splitting Gibbs measures (TISGMs) for the HC-Blume-Capel model on a "wand" graph embedded in the Cayley tree of arbitrary order $k \geq 2$. It is known that there is the exact critical value $\theta_{cr}$ such that for $\theta \geq \theta_{cr}$, there exists a unique TISGM, whereas for $\theta < \theta_{cr}$, precisely three TISGMs exist in the case of a "wand" graph for the given model. In the present paper, we completely solve the (non-)extremality problem for one of these measures for any order $k$ of Cayley tree.}




\maketitle
\section{Introduction}\label{intro}
A central goal in statistical physics is to determine all Gibbs measures for a given Hamiltonian, as these describe the probability distribution of microscopic states in a physical system. Each limiting Gibbs measure corresponds to a distinct phase, and a phase transition occurs when their number changes. It is well known that the set of limiting Gibbs measures forms a non-empty, convex, compact subset of probability measures, and each such measure can be decomposed into its extreme points, i.e., extreme Gibbs measures \cite{1,2,3,4}.

The hard-core (HC) model, introduced by Mazel and Suhov on the $d$-dimensional lattice $\mathbb{Z}^d$~\cite{5}, has been extensively studied, especially on Cayley trees \cite{6,7,8,9,10,11}. In particular, \cite{9} offers a full characterization of weakly periodic Gibbs measures for the HC model under a normal divisor of index $2$. In \cite{12}, fertile HC models were introduced based on interaction graphs such as the \emph{loop}, \emph{whistle}, \emph{wand}, and \emph{key}. Further studies \cite{13,14,15,16} analyzed fertile HC models with three states on Cayley trees of order $k > 1$. Specifically, \cite{14,15} classified all translation-invariant splitting Gibbs measures (TISGMs) for the \emph{loop} and \emph{wand} cases when $k=2$ and $k=3$, respectively, while \cite{16} showed that at least three TISGMs exist for any $k$ and addressed their extremality for $k=2$.

This work closely relates to prior studies on phase transitions and Gibbs measures on Cayley trees, particularly in systems with hard-core constraints and multiple spin states. Foundational techniques for analyzing TISGMs and their extremality, especially for constrained models, are developed in \cite{17,18,19,20}. Although the Blume-Capel model has been widely studied on both Cayley trees and lattice structures \cite{21,22}, our focus is a generalized version-the HC-Blume-Capel model with additional hard-core constraints, particularly for the \emph{wand}-type interaction graph on Cayley trees of arbitrary order \( k \geq 2 \). Prior studies mostly addressed either the unconstrained Blume-Capel model or cases with small tree orders (e.g., \( k = 2 \)).

The Blume-Capel model, a spin system where each spin takes values $-1, 0, 1$, was originally proposed to model $^3$He-$^4$He phase transitions \cite{23}. The state $\sigma(x) = 0$ denotes a vacancy at site $x$, while $\sigma(x) = \pm 1$ indicates occupied states with spin $\pm 1$ \cite{23,24,25}. Extensive results exist on this model \cite{26,27,28,29,30,31,32,33}. In \cite{27}, TISGMs for the HC-Blume-Capel model with a \emph{wand} graph were studied for $k=2$, identifying a critical value $\theta_{cr} = 1$ such that a unique TISGM exists for $\theta \geq \theta_{cr}$, and exactly three for $\theta < \theta_{cr}$. The role of chemical potential and extremality of these measures were further addressed in \cite{30}.

In \cite{33}, the results of \cite{27} are extended to arbitrary tree orders $k \geq 2$ for the \emph{wand}-type graph, where the exact critical value
\begin{equation} \nonumber
    \theta_{\text{cr}} = \sqrt[k+1]{\frac{k^k (k-1)}{2^k}}
\end{equation}
is found and also it is shown that for $\theta \geq \theta_{\text{cr}}$, there is a unique TISGM, while for $\theta < \theta_{\text{cr}}$, three TISGMs exist. The main result of the present paper is to study the (non)extremality problem for one of these measures for any $k$ using both the Kesten-Stigum and Martinelli-Sinclair-Weitz criteria.

 Our main results can be summarized as follows: for the case \(k = 3\), we show that the measure \(\mu_0\) is non-extremal for \(\theta \in (0, \theta_c^{(1)}) \cup (\theta_c^{(2)}, +\infty)\), where \(\theta_c^{(1)} \approx 0.83\) and \(\theta_c^{(2)} \approx 1.226\) (see Theorem~\ref{thm3}). On the other hand, for the intermediate range \(\theta_c^{(1)} < \theta < \theta_c^{(2)}\), the measure \(\mu_0\) is extremal (see Theorem~\ref{thm6}). For all larger values \(k \geq 4\), we prove that \(\mu_0\) is non-extremal for all \(\theta > 0\) (see Theorem~\ref{thm5}).

The paper is structured as follows. Section~\ref{Preliminaries} presents preliminary concepts and the system of functional equations, derives the critical value \( \theta_{\text{cr}} \) and the corresponding number of TISGMs. Sections~\ref{Conditions} and~\ref{kthree} address the (non)extremality of the measure \( \mu_0 \). Section~\ref{conclusion} summarizes key findings and outlines future research directions.

\section{Preliminaries}\label{Preliminaries}

A Cayley tree $\Gamma^{k} = (V, L)$ of order $k \geq 1$ is an infinite tree, i.e., a graph without cycles, where each vertex has exactly $k+1$ adjacent edges. Here, $V$ represents the set of vertices of $\Gamma^{k}$, and $L$ denotes its set of edges. The vertices $x$ and $y$ are called \emph{nearest neighbors} iff $\langle x,y \rangle\in L .$ The distance on this tree is defined as the number of nearest neighbour pairs of the minimal path between the vertices $x$ and $y$ (where path is a collection of nearest neighbour pairs, two consecutive pairs sharing at least a given vertex) and denoted by $d(x,y).$

Consider a model where the spin takes values from the set $\Phi = \{-1, 0, 1\}$. Then, a \emph{configuration} $\sigma$ on $V$ is defined as a function $x \in V \mapsto \sigma(x) \in \Phi$. The set of all possible configurations is given by $\Omega = \Phi^V$. For a subset $A \subset V$, let $\Omega_A$ denote the space of configurations defined on $A$.

The Hamiltonian of the Blume-Capel model is given by:
\begin{equation} \label{f1}
    H(\sigma) = J \sum_{\langle x, y \rangle, x, y \in V} (\sigma(x) - \sigma(y))^2,
\end{equation}
where $J \in \mathbb{R}$.

Let $x^0 \in V$ be a fixed vertex. We write $x \prec y$ if the shortest path from $x^0$ to $y$ passes through $x$. Define the sets:
\begin{equation} \nonumber
    W_n = \{x \in V : d(x^0, x) = n\}, \quad V_n = \{x \in V : d(x^0, x) \leq n\}.
\end{equation}
A vertex $y$ is called a \emph{direct successor} of $x$ if $x \prec y$ and $d(x, y) = 1$. Let $S(x)$ denote the set of direct successor of $x \in V$.

Let $G$ be a graph with vertices $\{-1,0,1\}$ and the set of edges is equal to $\{0,-1\},$ $\{0,1\},$ $\{-1,-1\},$ $\{1,1\}.$ We call $G$ a graph of type "wand."

A configuration $\sigma$ is called a $G$-\emph{admissible configuration} on the Cayley tree if for any nearest-neighbor pair $x, y$ from $V$, the pair $\{\sigma(x), \sigma(y)\}$ is an edge in $G$. The set of $G$-admissible configurations is denoted by $\Omega^G.$ The restriction of this set on subsets $A\subset V$ is denoted by $\Omega^G_A.$

We now give a construction of the splitting Gibbs measures. Let $h: x \mapsto h_x = (h_{-1,x}, h_{0,x}, h_{1,x})$ be a vector function for $x \in V \setminus \{x^0\}$. Consider the probability distribution $\mu^{(n)}$ on $\Omega_{V_n}^G$ given by:
\begin{equation} \label{f2}
    \mu^{(n)}(\sigma_n) = Z_n^{-1} \exp \left\{-\beta H_n(\sigma_n) + \sum_{x \in W_n} h_{\sigma_n(x), x} \right\},
\end{equation}
where $\sigma_n \in \Omega_{V_n}^G$, $H_n$ is the Hamiltonian $H$ restricted to $V_n$ and the partition function is given by:
\begin{equation}\nonumber
    Z_n = \sum_{\overline{\sigma}_n \in \Omega_{V_n}^G} \exp \left\{-\beta H_n(\overline{\sigma}_n) + \sum_{x \in W_n} h_{\overline{\sigma}_n(x), x} \right\}.
\end{equation}

A probability distribution $\mu^{(n)}$ $(\forall n \geq 1)$ is said to be consistent if
\begin{equation} \label{f3}
    \sum_{\sigma^{(n)}} \mu^{(n)}(\sigma_{n-1}, \sigma^{(n)}) = \mu^{(n-1)}(\sigma_{n-1})
\end{equation}
for all $n \geq 1$ and $\sigma_{n-1} \in \Omega_{V_{n-1}}^G$.

In this case, due to Kolmogorov's extension theorem, there exists a unique measure $\mu$ on $\Omega_V^G$ such that:
\begin{equation} \nonumber
    \mu(\{\sigma \mid_{V_n} = \sigma_n\}) = \mu^{(n)}(\sigma_n),
\end{equation}
for all $n \geq 1$ and $\sigma_n \in \Omega_{V_n}^G$.

Let $L(G)$ denote the set of edges of the graph $G$, and let $A \equiv A^G = (a_{ij})_{i,j=-1,0,1}$ represent the adjacency matrix of $G$, i.e.,
\begin{equation*}
    a_{ij} \equiv a_{ij}^G =
    \begin{cases}
        1, & \text{if } \{i,j\} \in L(G), \\
        0, & \text{if } \{i,j\} \notin L(G).
    \end{cases}
\end{equation*}

The following theorem provides necessary and sufficient conditions on $h_{i,x}$ under which Equation \eqref{f3} holds.

\begin{theorem}\label{thm1} \cite{27} Let $k \geq 2$. The probability distribution $\mu^{(n)}, n=1,2,\dots$ in \eqref{f2} is consistent if and only if, for any $x \in V$, the following conditions hold:
\begin{equation} \label{f4}
    \begin{cases}
        z_{1,x} = \prod\limits_{y \in S(x)} \frac{a_{1,-1} \theta^4 z_{-1,y} + a_{1,0} \theta + a_{1,1} z_{1,y}}{a_{0,-1} \theta z_{-1,y} + a_{0,0} + a_{0,1} \theta z_{1,y}}, \\
        z_{-1,x} = \prod\limits_{y \in S(x)} \frac{a_{-1,-1} z_{-1,y} + a_{-1,0} \theta + a_{-1,1} \theta^4 z_{1,y}}{a_{0,-1} \theta z_{-1,y} + a_{0,0} + a_{0,1} \theta z_{1,y}},
    \end{cases}
\end{equation}
where $\theta = \exp(-J \beta), \beta = 1/T$, and $z_{i,x} = \exp(h_{i,x} - h_{0,x})$ for $i=-1,1$. \end{theorem}

Translation-invariant splitting Gibbs measures (TISGMs) correspond to solutions of Equation (4) with $z_{i,x} = z_i$ for all $x \in V$ and $i=-1,1$. For convenience, we denote $z_{+1} = z_1$ and $z_{-1} = z_2$. Then, in the case of a \emph{wand-type} graph, the system of equations \eqref{f4} simplifies to:
\begin{equation} \label{f5}
    \begin{cases}
        z_1 = \left(\frac{\theta + z_1}{\theta z_1 + \theta z_2}\right)^k, \\
        z_2 = \left(\frac{\theta + z_2}{\theta z_1 + \theta z_2}\right)^k.
    \end{cases}
    \tag{5}
\end{equation}

First, consider the case $z_1 = z_2 = z$. Then, from Equation \eqref{f5}, we obtain
\begin{equation} \label{f6}
    z = \left(\frac{\theta + z}{2 \theta z}\right)^k. \tag{6}
\end{equation}

The following lemma holds.

\begin{lemma}\label{lem1} \textit{Let $k \geq 2$. The following assertions hold
\begin{enumerate}
\item For any $\theta > 0$, Eq. \eqref{f6} has a unique positive solution.
\item Let $z^*$ be the unique positive solution of equation \eqref{f6}. Then
\begin{enumerate}
    \item If $\theta>1$, then $z^*<\theta$.
    \item If $\theta<1$, then $z^*>\theta$.
\end{enumerate}
\end{enumerate}}
\end{lemma}

\begin{proof} (1) Rewrite Eq. \eqref{f6} as
\begin{equation} \nonumber
    z = f(z),
\end{equation}
where
\begin{equation*}
    f(z) = \left(\frac{\theta + z}{2 \theta z}\right)^k.
\end{equation*}
Differentiating $f(z)$, we obtain
\begin{equation*}
    f'(z) = k \left(\frac{\theta + z}{2 \theta z}\right)^{k-1} \left(-\frac{1}{2z^2}\right) < 0,
\end{equation*}
indicating that $f$ is a decreasing function for $z > 0$. Hence, equation $z = f(z)$ has a unique solution $z^* = z^*(k, \theta)$ for any $\theta > 0$.

(2.a). Let $\theta>1$ and $z^*=f(z^*),$  then we need to prove $z^*<\theta$. Assume the opposite, i.e., $z^*>\theta$. Since $f(z)$ is decreasing, then $z^*=f(z^*)<f(\theta)=\frac{1}{\theta^{k}}<1$. This is a contradiction, since $z^*>\theta>1$. Thus, the first assertion of the lemma is proved. The remaining part is proved similarly. The lemma is proved.\end{proof}

The following theorem holds.

\begin{theorem}\label{thm2}\cite{33} Let $k \geq 2$ and define
\begin{equation*}
    \theta_{cr} = \theta_{cr}(k) = \sqrt[k+1]{\frac{k^k (k-1)}{2^k}}.
\end{equation*}
Then, for the HC-Blume-Capel model in the case of a wand-type graph:
\begin{enumerate}
    \item If $\theta \geq \theta_{cr}$, there exists exactly one TISGM, denoted $\mu_0$.
    \item If $\theta < \theta_{cr}$, there exist exactly three TISGMs, denoted $\mu_0, \mu_1, \mu_2$.
\end{enumerate} \end{theorem}

\begin{proof} The proof is given in \cite{33}, however, with some typos. \end{proof}

\begin{remark}\label{rk0} In general, one would expect an $\theta-$dependent uniqueness regime for the set of Gibbs measures corresponding to this model, for example around $\theta=\exp\{-\beta J\}=1$ where the quadratic interaction vanishes, however, due to the hard-core constraints there are multiple Gibbs measures.\end{remark}

\section{Conditions for Non-Extremality of the Measure $\mu_{0}$}\label{Conditions}

To analyze the (non-)extremality of the measure, we employ the method presented in \cite{34,35,36}. Specifically, we consider Markov chains with state space $\{-1,0,1\}$ and define the probability transition matrix $P_{\mu}$ associated with the given TISGM $\mu$. The transition probability $P_{\sigma(x)\sigma(y)}$ represents the probability of transitioning from state $\sigma(x)$ to state $\sigma(y)$.

A sufficient condition for the corresponding Gibbs measure to be non-extremal is given by the Kesten-Stigum criterion:
\begin{equation}\nonumber
    k\lambda_{2}^{2} > 1,
\end{equation}
where $\lambda_{2}$ is the second largest eigenvalue of the matrix $P_{\mu}$.

For $k \geq 2$ and $\theta \geq \theta_{cr},$ it is evident that the system of equations \eqref{f5} has a unique solution $(z^*, z^*),$ where $z^*$ is the unique solution of the equation
\begin{equation}\label{f20}
    z = \left(\frac{\theta+z}{2\theta z}\right)^k. \tag{20}
\end{equation}
For $0 < \theta < \theta_{cr}$, the system admits three solutions: $(z^*,z^*)$, $(z_1,z_2)$, and $(z_2,z_1)$.

Let us determine the conditions for the non-extremality of the measures corresponding to these solutions. To this end, we consider the transition probabilities given by
\begin{equation}\nonumber
P_{\sigma (x) \sigma (y)}=\frac{a_{\sigma(x), \sigma(y)}\exp\{-J\beta (\sigma (x)- \sigma (y))^2+h_{\sigma (y)}\}}{\sum_{\widetilde{\sigma}(y)\in\{-1, 0, 1\}}a_{\sigma(x), \widetilde{\sigma}(y)}\exp\{-J\beta (\sigma(x)-\widetilde{\sigma}(y))^2 +h_{\widetilde{\sigma}(y)}\}}.
\end{equation}

In the case under consideration, where \( G \) is a \emph{wand}, the coefficients \( a_{\sigma(x),\sigma(y)} \) are given by
\begin{align}\nonumber
    a_{-1, -1} &= 1, \quad a_{-1, 0} = 1, \quad a_{-1, 1} = 0, \notag \\
    a_{0, -1} &= 1, \quad a_{0, 0} = 0, \quad a_{0, 1} = 1, \notag \\
    a_{1, -1} &= 0, \quad a_{1, 0} = 1, \quad a_{1, 1} = 1. \notag
\end{align}

From these values, we obtain the transition probabilities
\begin{align} \nonumber
    P_{-1, -1} &= \frac{z_2}{z_2+\theta}, & P_{-1, 0} &= \frac{\theta}{z_2+\theta}, & P_{-1, 1} &= 0, \notag \\
    P_{0, -1} &= \frac{z_2}{z_1+z_2}, & P_{0, 0} &= 0, & P_{0, 1} &= \frac{z_1}{z_1+z_2}, \notag \\
    P_{1, -1} &= 0, & P_{1, 0} &= \frac{\theta}{z_1+\theta}, & P_{1, 1} &= \frac{z_1}{z_1+\theta}. \notag
\end{align}

Consequently, the transition matrix takes the form
\begin{equation}
    P=
    \begin{bmatrix}
        \frac{z_2}{z_2+\theta} & \frac{\theta}{z_2+\theta} & 0 \\
        \frac{z_2}{z_1+z_2} & 0 & \frac{z_1}{z_1+z_2} \\
        0 & \frac{\theta}{z_1+\theta} & \frac{z_1}{z_1+\theta}
    \end{bmatrix}.\tag{21}
\end{equation}

For a single solution of the system of equations \eqref{f5}, the matrix \( P \) simplifies to
\begin{equation}\label{f22}
   P=
    \begin{bmatrix}
        \frac{z}{z+\theta} & \frac{\theta}{z+\theta} & 0 \\
        \frac{1}{2} & 0 & \frac{1}{2} \\
        0 & \frac{\theta}{z+\theta} & \frac{z}{z+\theta}
    \end{bmatrix}, \tag{22}
\end{equation}
where \( z \) is the unique positive solution of equation \eqref{f20}.

It is evident that one of the eigenvalues of this matrix is \( s_3=1 \). The remaining two eigenvalues are
\begin{equation} \nonumber
    s_1 = \frac{z}{z+\theta}, \quad s_2 = -\frac{\theta}{z+\theta}.
\end{equation}

To determine the dominant eigenvalue, we compute $\max\{\lvert s_1\rvert, \lvert s_2\rvert\}$:
\begin{equation}\nonumber
     \left\lvert s_1\right\rvert - \left\lvert s_2\right\rvert = \frac{z}{z+\theta} - \frac{\theta}{z+\theta} = \frac{z-\theta}{z+\theta}.
\end{equation}

Let $0<\theta<1$. From Lemma \ref{lem1}, it follows that
\begin{equation}\nonumber
    \max\{\mid s_1\mid, \mid s_2\mid\} = \mid s_1\mid.
\end{equation}

For $\theta>1$, Lemma \ref{lem1} implies that
\begin{equation}\nonumber
    \max\{\mid s_1\mid, \mid s_2\mid\} = \mid s_2\mid,
\end{equation}
i.e.,
\begin{equation}\nonumber
    \max\{\mid s_1\mid, \mid s_2\mid\} = \begin{cases}
        \mid s_1\mid, & \text{if } 0<\theta<1, \\
        \mid s_2\mid, & \text{if } \theta>1.
    \end{cases}
\end{equation}
Consequently, for $\theta > 1$, the inequality holds:
\begin{equation}\nonumber
    s_1 < \mid s_2\mid < s_3 = 1.
\end{equation}

Now, let $\theta > 1$ and set $k=3$. In this case, we examine the Kesten-Stigum condition for the non-extremality of the measure $\mu_0$, which requires $3s_2^2 > 1$. By employing the Ferrari method, we solve equation \eqref{f20}, which admits a real solution:
\begin{equation}\label{f23}
    z = \frac{\sqrt{A(\theta)} + \frac{1}{16\theta^3} + \sqrt{\left(\sqrt{A(\theta)} + \frac{1}{16\theta^3}\right)^2 - 4\left(\frac{y(\theta)}{2} - \sqrt{C(\theta)}\right)}}{2}, \tag {23}
\end{equation}
where
\begin{equation}\nonumber
    A(\theta) = \frac{1}{256\theta^6} + \frac{3}{8\theta^2} + y(\theta), \quad C(\theta) = \frac{y(\theta)^2}{4} + \frac{1}{8},
\end{equation}

\begin{equation} \nonumber
    y(\theta) = \frac{\frac{1}{24} \sqrt[3]{108\theta^4 + 12\sqrt{6144\theta^{12} + 81\theta^8}} - \frac{4\theta^4}{\sqrt[3]{108\theta^4 + 12\sqrt{6144\theta^{12} + 81\theta^8}}} - \frac{1}{8}}{\theta^2}.
\end{equation}

To determine the interval of non-extremality for this measure, we verify the condition:
\begin{equation}\nonumber
    3s_2^2 - 1 = 3 \left( \frac{\theta}{z+\theta} \right)^2 - 1 > 0,
\end{equation}
where $z$ is given by \eqref{f23}. Using the Maple software, we confirm that this inequality holds for $\theta \in (\theta_c^{(2)}, +\infty)$, indicating that under this condition, the measure $\mu_0$ is non-extremal. Here, $\theta_c^{(2)} \approx 1.226$ (see Fig. 2).

\begin{center}
\includegraphics[width=6cm]{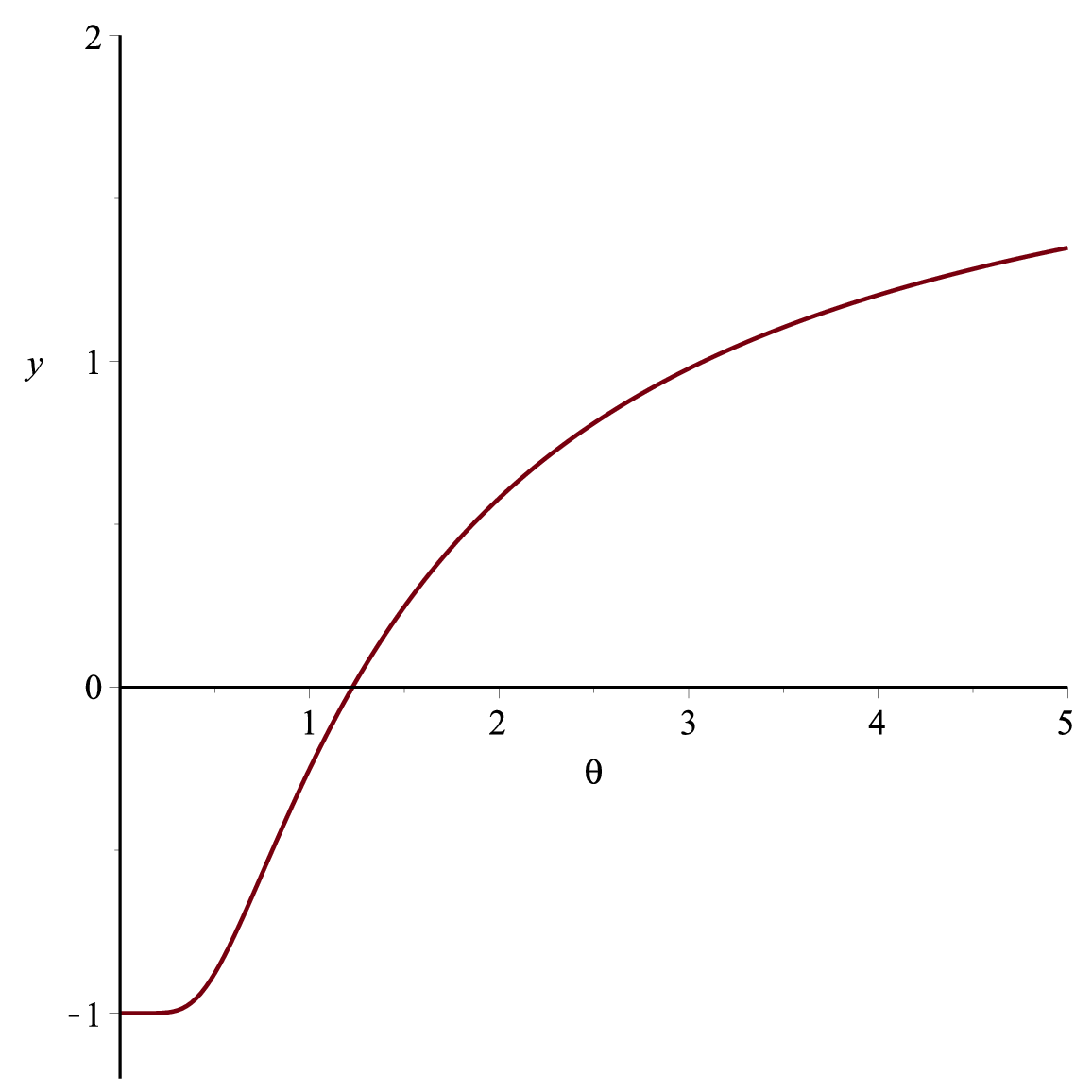}\includegraphics[width=6cm]{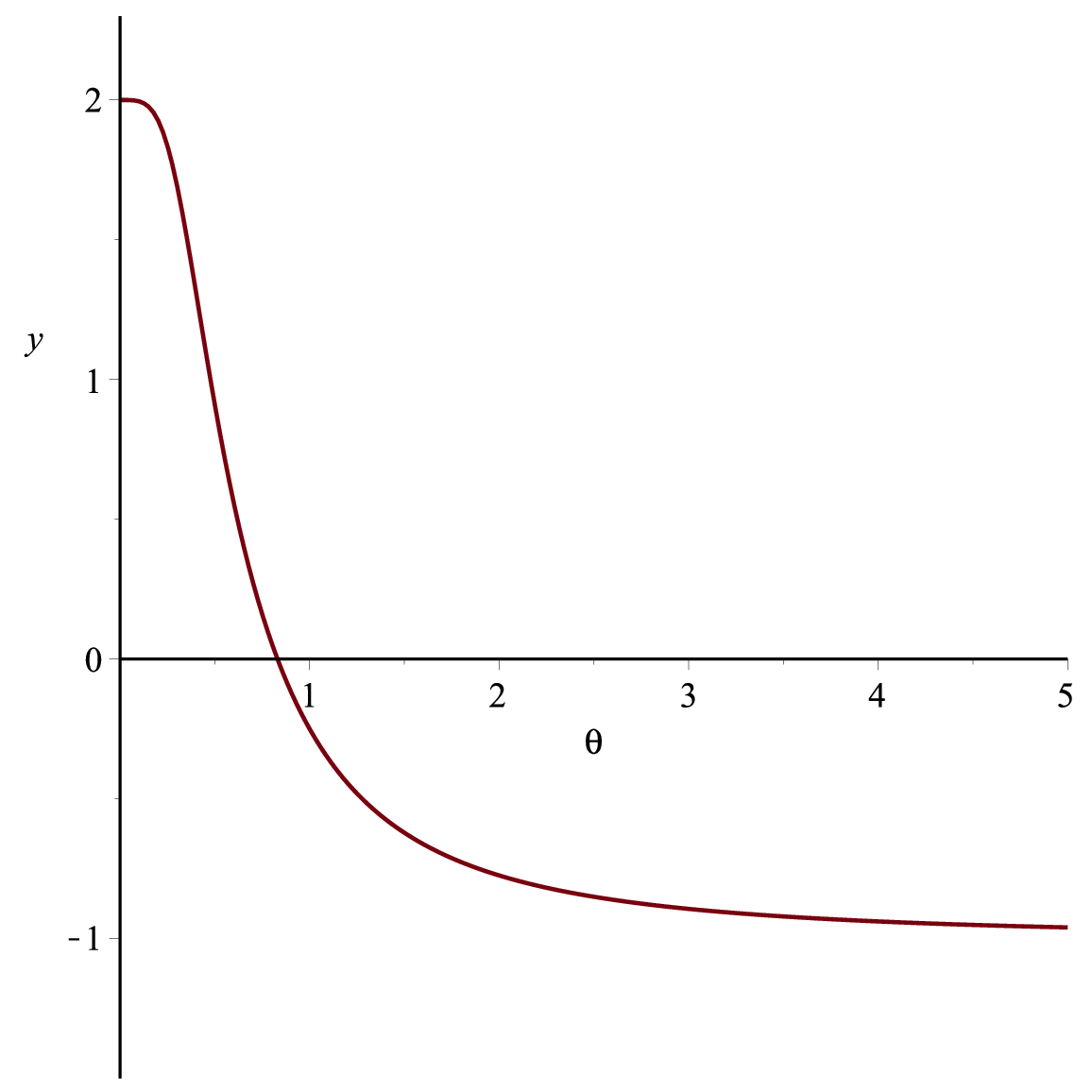}\label{fig2}
\end{center}
\begin{center}{\footnotesize \noindent
 Figure 2. The graph of functions $3s_2^2-1$ (left) and $3s_1^2-1$ (right).}\
\end{center}

Now, let \( 0 < \theta < 1 \). In this case, we have \( \mid s_2\mid < s_1 < s_3 = 1 \). The Kesten-Stigum condition for the non-extremality of the measure \( \mu_0 \) is given by

\[
3s_1^2 > 1.
\]

To determine the interval of non-extremality of \( \mu_0 \) for \( 0 < \theta < 1 \), we examine the inequality

\[
3s_1^2 - 1 = 3\left(\frac{z}{z+\theta}\right)^2 - 1 > 0,
\]
where \( z \) is given by equation \eqref{f23}. Using the Maple program, we verify that this inequality holds for \( \theta \in (0, \theta_{c}^{(1)}) \). Thus, under this condition, the measure \( \mu_0 \) is non-extremal, where \( \theta_{c}^{(1)} \approx 0.83 \) (see Fig. 2).

\begin{theorem}\label{thm3} Let $k=3$. Then, for the HC-Blume-Capel model on a \emph{wand} graph, the measure $\mu_0$ is not extremal for $\theta \in (0, \theta_{c}^{(1)}) \cup (\theta_{c}^{(2)}, +\infty)$, where $\theta_{c}^{(1)} \approx 0.83$ and $\theta_{c}^{(2)} \approx 1.226$. \end{theorem}

For $k=2$, the following result was established in \cite{27}:

\begin{remark}\label{rk2}. Let $k=2$,  then for the HC-Blume-Capel model on a \emph{wand} graph, the measure $\mu_0$ is not extremal for $\theta \in \left(0, \frac{1}{2}\sqrt[3]{4\sqrt{2}-4}\right) \cup \left(\frac{1}{2}\sqrt[3]{28+20\sqrt{2}}, +\infty\right)$, where $\frac{1}{2}\sqrt[3]{4\sqrt{2}-4} \approx 0.591$ and $\frac{1}{2}\sqrt[3]{28+20\sqrt{2}} \approx 1.915$. \end{remark}

\begin{theorem}\label{thm5} Let $k \geq 4$. Then, for the HC-Blume-Capel model on a \emph{wand} graph, the measure $\mu_0$ is not extremal for any $\theta > 0$. \end{theorem}

\begin{proof} Let $0 < \theta < 1$. In this case, we have
\[
s_{1} = \frac{z}{z+\theta}.
\]
The Kesten-Stigum condition takes the form
\[
k \cdot \left(\frac{\frac{z}{\theta}}{1+\frac{z}{\theta}}\right)^{2} > 1.
\]
By Lemma \ref{lem1}, for $\theta \in (0,1)$, we obtain $\frac{z}{\theta} > 1$. Since the function $f(x) = \frac{x}{1+x}$ is increasing, for $x > 1$ we have  $f(x)>f(1)=\frac{1}{2}.$ Then, it follows that
\[
k \cdot \left(\frac{\frac{z}{\theta}}{1+\frac{z}{\theta}}\right)^{2} > \frac{k}{4} \geq 1.
\]

Now, consider $\theta > 1$. In this case, we have
\[
\mid s_{2}\mid=\frac{\theta}{z+\theta}.
\]
The Kesten-Stigum condition now takes the form
\[
k \cdot \left(\frac{1}{1+\frac{z}{\theta}}\right)^{2} > 1.
\]
For $\theta > 1$, Lemma \ref{lem1} implies $\frac{z}{\theta} < 1$. Since the function $f(x) = \frac{1}{1+x}$ is decreasing, for $x < 1$ we have $f(x)>f(1)=\frac{1}{2}$. Thus, we obtain
\[
k \cdot \left(\frac{1}{1+\frac{z}{\theta}}\right)^{2} > \frac{k}{4} \geq 1.
\]
This inequality holds for $k \geq 4$, which implies that $\mu_{0}$ is not extremal for $k \geq 4$. The theorem is therefore proved.\end{proof}

\begin{remark}\label{rk2} Since the set of all Gibbs measures is convex-compact, it shows that for $k\geq4$ there exist non-translation-invariant extremal splitting Gibbs measures.\end{remark}

\section{Conditions for the Extremality of the Measure $\mu_0$ for $k=3$}\label{kthree}

The methods presented in \cite{35} are well-known for studying extremality. We begin by recalling the necessary definitions from \cite{35}. Removing an arbitrary edge $\langle x^{0},x^{1}\rangle = l \in L$ from the Cayley tree $\Gamma^{k}$ results in its division into two disjoint components, $\Gamma_{x^{0}}^{k}$ and $\Gamma_{x^{1}}^{k}$, each referred to as a semi-infinite tree or a semi-Cayley tree.

Consider a finite, rooted subtree $\Im \subset \Gamma^k$ such that its root coincides with the root of the half-tree $\Gamma^k_{x^0}$, and all vertices in $\Im$ lie within a finite depth from $x_0$. We assume that $\Im$ includes all vertices up to some fixed generation $n$ of the half-tree. The boundary $\partial\Im$ of the subtree $\Im$ consists of the nearest neighbors of its vertices that belong to $\Gamma_{x^{0}}^{k} \setminus \Im$. We identify subgraphs of $\Im$ with their set of vertices $A$, and their boundary by $\partial A$. We also mention that the boundary is a subset of $\left( \Im\cup \partial \Im \right)\backslash A.$  In \cite{35}, two key quantities, $\kappa$ and $\gamma$, are introduced, which play a crucial role in the study of the extremality of TISGMs. These quantities characterize the set of Gibbs measures $\mu_{\Im_x}^{\tau}$, where the boundary condition $\tau$ is fixed.

For a given rooted subtree $\Gamma_{x^{0}}^{k}$ and a vertex $x \in \Im$, we denote by $\Im _{x}$ the maximal subtree of $\Im$ with $x$ as its initial vertex. If $x$ is not the initial vertex of $\Im$, we define $\mu_{\Im_x}^{s}$ as the Gibbs measure in which the ``ancestor'' of $x$ has spin $s$, while the configuration at the lower boundary of $\Im_{x}$ (i.e., on $\partial\Im \setminus \{\text{ancestor of } x\}$) is given by $\omega$.

For two probability measures on $\Omega$, we define their distance by the norm:
\begin{equation}\nonumber
    \|\mu_1 - \mu_2 \|_x = \frac{1}{2} \sum_{i \in \{-1, 0, 1\}} \mid  \mu_1(\sigma(x) = i) - \mu_2(\sigma(x) = i) \mid.
\end{equation}

Let $\eta^{x, s}$ denote a configuration $\eta$ with spin at $x$ fixed at $s$. Following \cite{35}, we define:
\begin{equation}\nonumber
    \kappa \equiv \kappa(\mu) = \sup_{x \in \Gamma^k} \max_{x, s, s'} \|\mu_{\Im_x}^s - \mu_{\Im_x}^{s'}\|_x,
\end{equation}
\begin{equation}\nonumber
    \gamma \equiv \gamma(\mu) = \sup_{A \subset \Gamma^k} \max \|\mu_A^{\eta^{y, s}} - \mu_A^{\eta^{y, s'}}\|_x,
\end{equation}
where the second maximum is taken over all boundary conditions $\eta$, all $y \in \partial A$, all neighbors $x \in A$ of vertex $y$, and all spins $s, s' \in \{-1, 0, 1\}$.

To determine the condition for the extremality of the measure $\mu_0$ for $k=3$, we note from Lemma \ref{lem1} that $z > \theta$ for $0 < \theta < 1$ and $z < \theta$ for $\theta > 1$.

The parameter $\kappa$ has the particularly simple form \cite{35}:
\begin{equation}\nonumber
    \kappa = \frac{1}{2} \max_{i,j}\sum_{l \in \{-1, 0, 1\}} \mid  P_{il} - P_{jl}\mid .
\end{equation}

From this, it follows that $\kappa=0$ for $i = j$. For $i \neq j$, we compute
\begin{equation}\nonumber
    \sum_{l \in \{-1, 0, 1\}} \mid  P_{il} - P_{jl} \mid  =
    \begin{cases}
        \ \ 1, & \text{if } 0 < \theta < 1 \text{ and } ij = 0, \\
       \frac{2\theta}{z+\theta}, & \text{if } \theta \geq 1 \text{ and } ij = 0, \\
        \frac{2z}{z+\theta}, & \text{if } \theta > 0 \text{ and } ij = -1.
    \end{cases}
\end{equation}

Thus, for solutions $(z^*, z^*)$ with $0 < \theta < 1$, we obtain:
\begin{equation}\nonumber
    \kappa = \frac{z}{z+\theta},
\end{equation}
while for $\theta \geq 1$, we have:
\begin{equation}\nonumber
    \kappa = \frac{\theta}{z+\theta}.
\end{equation}

We now seek an estimate for $\gamma$, following the approach of \cite{35}, p. 15:
\begin{equation}\nonumber
    \gamma = \max \{ \|\mu_A^{\eta^{y,-1}} - \mu_A^{\eta^{y,0}}\|_x, \|\mu_A^{\eta^{y,-1}} - \mu_A^{\eta^{y,1}}\|_x, \|\mu_A^{\eta^{y,0}} - \mu_A^{\eta^{y,1}}\|_x \}.
\end{equation}

To estimate the constant $\gamma(z_{i},z_{i})$ depending on the boundary law labeled by $i$, for our model, we prove several lemmas.

\begin{lemma}\label{lem2} \textit{Recall the matrix $P$ given by formula \eqref{f22}, and let $\mu = \mu(\theta)$ be the corresponding Gibbs measure. Then, for any subset $A \subset \Im$ (where $\Im$ is the initial full subtree $\Gamma^k$), any boundary configuration $\eta$, any pair of spins $(s_1, s_2)$, any vertex $y \in \partial A$, and any neighbor $x \in A$ of $y$, we have
\begin{equation}\nonumber
\|\mu_A^{\eta^{y, s_{1}}} - \mu_A^{\eta^{y, s_{2}}}\|_x\leq\{\mid p^{0}(0)-p^{2}(0)\mid, \mid p^{0}(1)-p^{1}(1)\mid, \mid p^{0}(2)-p^{1}(2)\mid, \mid p^{0}(0)-p^{1}(0)\mid\},
\end{equation}
where $p^{t}(s):=\mu_{A}^{\eta^{y,t}}(\sigma(x)=s).$}
\end{lemma}

\begin{proof} Let $p_{s}=\mu_{A}^{\eta^{y,free}}(\sigma(x)=s), s=0,1.$ By the definition of the matrix $P$, we have
\begin{equation}\nonumber
p^{0}(0)=\frac{zp_{0}}{zp_{0}+\theta p_{1}},\ \ p^{0}(1)=\frac{\theta p_{1}}{zp_{0}+\theta p_{1}}, \ \ p^{0}(2)=0;
\end{equation}
\begin{equation}\label{f24}
p^{1}(0)=\frac{1}{2},\ \ p^{1}(1)=0, \ \ p^{1}(2)=\frac{1}{2};\tag{24}
\end{equation}
\begin{equation}\nonumber
p^{2}(0)=0,\ \ p^{2}(1)=\frac{\theta p_{1}}{zp_{0}+\theta p_{1}}, \ \ p^{2}(2)=\frac{zp_{0}}{zp_{0}+\theta p_{1}},
\end{equation}
and hence the proposition follows from the following Lemmas 3 and 4.
\end{proof}

Clearly,
\[
\text{a) } \theta<1 \Rightarrow z>\theta; \qquad
\text{b) } \theta>1 \Rightarrow z<\theta.
\]
Let $p_0\ge0$, $p_1\ge0$, $p_0+p_1=1$, and define the probability vectors
\begin{equation}\nonumber
p^{0}=\left(\frac{z p_{0}}{z p_{0}+\theta p_{1}},\ \frac{\theta p_{1}}{z p_{0}+\theta p_{1}},\ 0\right), \quad
p^{1}=\left(\frac12,\ 0,\ \frac12\right),
\end{equation}
\begin{equation}\nonumber
p^{2}=\left(0,\ \frac{\theta p_{1}}{z p_{0}+\theta p_{1}},\ \frac{z p_{0}}{z p_{0}+\theta p_{1}}\right).
\end{equation}

\begin{lemma}\label{lem3} \textit{Let
\begin{equation}\nonumber
A:=\frac{z p_{0}}{z p_{0}+\theta p_{1}},\qquad B:=1-A=\frac{\theta p_{1}}{z p_{0}+\theta p_{1}}.
\end{equation}
Then
\begin{equation}\nonumber
\max_{i,j,k}\lvert p^{i}(k)-p^{j}(k)\rvert=\max\{A,\ 1-A\}.
\end{equation}
Consequently, the value $1$ is attained if and only if $A\in\{0,1\}$, i.e., when $p_0\in\{0,1\}$. Therefore, for any choice
\begin{equation}\nonumber
0<p_0<1,\qquad p_1=1-p_0\in(0,1),
\end{equation}
we have $\max_{i,j,k}\lvert p^{i}(k)-p^{j}(k)\rvert<1$. Moreover, the minimal possible value of this maximum equals $\frac{1}{2}$ and is attained at
\begin{equation}\nonumber
p_0=\frac{\theta}{z+\theta},\qquad p_1=\frac{z}{z+\theta},
\end{equation}
when $A=\frac{1}{2}$.}
\end{lemma}

\begin{proof}
We can write:
\begin{equation}\nonumber
p^0=(A,1-A,0),\quad p^1=\left(\tfrac12,0,\tfrac12\right),\quad p^2=(0,1-A,A).
\end{equation}
The coordinate differences between these three vectors (for all $i,j\in\{0,1,2\}$ and $k\in\{0,1,2\}$) give the set of values
\begin{equation}\nonumber
\{\,A,\ 1-A,\ \tfrac12,\ \mid A-\tfrac12\mid\,\}.
\end{equation}
Indeed:
\begin{equation}\nonumber
\begin{aligned}
&\mid p^0(0)-p^2(0)\mid=A,\quad \mid p^0(2)-p^2(2)\mid=A,\quad \mid p^0(1)-p^2(1)\mid=0,\\
&\mid p^0(1)-p^1(1)\mid=1-A,\quad \mid p^2(1)-p^1(1)\mid=1-A,\\
&\mid p^0(2)-p^1(2)\mid=\tfrac12,\quad \mid p^1(0)-p^2(0)\mid=\tfrac12,\\
&\mid p^0(0)-p^1(0)\mid=\mid A-\tfrac12\mid,\quad \mid p^1(2)-p^2(2)\mid=\mid\tfrac12-A\mid.
\end{aligned}
\end{equation}
Since $\tfrac12\le \max\{A,1-A\}$ and $\mid A-\tfrac12\mid\le \max\{A,1-A\}$ for any $A\in[0,1]$, it follows that
\begin{equation}\nonumber
\max_{i,j,k}\mid p^i(k)-p^j(k)\mid=\max\{A,1-A\}.
\end{equation}
The maximum equals $1$ if and only if $A\in\{0,1\}$, which is equivalent to $p_0\in\{0,1\}$ (for positive $z,\theta$ ensuring a positive denominator $z p_0+\theta p_1$). Therefore, for any $p_0\in(0,1)$ (and $p_1=1-p_0$), we have $A\in(0,1)$ and hence $\max_{i,j,k}\mid p^i(k)-p^j(k)\mid<1$.

To minimize the maximum, we set $A=1-A$, i.e., $A=\tfrac12$. Solving
\begin{equation}\nonumber
\frac{z p_0}{z p_0+\theta(1-p_0)}=\frac12
\quad\Longleftrightarrow\quad
z p_0=\theta(1-p_0),
\end{equation}
gives
\begin{equation}\nonumber
p_0=\frac{\theta}{z+\theta},\quad p_1=1-p_0=\frac{z}{z+\theta}.
\end{equation}
For this choice $A=\tfrac12$ and thus
\begin{equation}\nonumber
\max_{i,j,k}\mid p^i(k)-p^j(k)\mid=\tfrac12.
\end{equation}
The lemma is proved.
\end{proof}

\begin{lemma}\label{lem4} \textit{Let $\theta>0$.
a) Then
\begin{equation}\label{f25}
\max\{A,1-A\}=A \tag{25}
\end{equation}
holds whenever
\begin{equation}\nonumber
p_{0}\geq\frac{\theta}{z+\theta}.
\end{equation}
b) Then
\begin{equation}\nonumber
\max\{A,1-A\}=1-A
\end{equation}
holds whenever
\begin{equation}\nonumber
p_{0}\leq\frac{\theta}{z+\theta}.
\end{equation}}
\end{lemma}

\begin{proof}
Assume $\theta>0$. First, we prove part a). From equality \eqref{f25} we get the inequality
\begin{equation}\label{f26}
\frac{z p_{0}}{(z-\theta)p_{0}+\theta}\geq\frac12. \tag{26}
\end{equation}
Let us solve \eqref{f26}. If $z\geq\theta$, then
\begin{equation}\nonumber
(z-\theta)p_{0}+\theta\geq\theta>0
\end{equation}
since $p_{0}\geq0$. If $z<\theta$, then
\begin{equation}\nonumber
(z-\theta)p_{0}+\theta\geq z>0
\end{equation}
since $p_{0}\leq1$ and $z>0$.

Thus, we can multiply \eqref{f26} by the denominator without changing the inequality sign:
\begin{equation}\nonumber
2z\,p_{0} \ \ge\ (z-\theta)\,p_{0}+\theta
\quad\Longleftrightarrow\quad
(z+\theta)\,p_{0} \ \ge\ \theta.
\end{equation}
This yields
\begin{equation}\nonumber
p_{0}\ \ge\ \frac{\theta}{z+\theta}.
\end{equation}
This proves part a). For part b) we solve the inequality
\begin{equation}\label{f27}
\frac{z p_{0}}{(z-\theta)p_{0}+\theta}\leq\frac12, \tag{27}
\end{equation}
which leads to the stated condition. The lemma is proved.
\end{proof}

\begin{proposition}\label{prop1}Regardless of the possible values $(z,z)$ (i.e., the unique solution of system \eqref{f5})

1) If $\theta>0$ and $p_{0}\geq\frac{\theta}{z+\theta}$, then
\begin{equation}\nonumber
\gamma(z,z)\leq\frac{z p_{0}}{(z-\theta)p_{0}+\theta};
\end{equation}

2) If $\theta>0$ and $p_{0}\leq\frac{\theta}{z+\theta}$, then
\begin{equation}\nonumber
\gamma(z,z)\leq\frac{\theta(1-p_{0})}{(z-\theta)p_{0}+\theta}.
\end{equation}
\end{proposition}

\begin{proof}
This follows from the above lemmas.
\end{proof}

\begin{remark}\label{rk3}
Computer analysis shows that if $p_{0}\in (0,\frac12)\cup(\frac12,1)$, then the extreme interval of the measure $\mu_{0}$ intersects the non-extreme interval. Therefore, only for $p_{0}=\frac12$ does it completely fill the remaining part of the non-extreme interval.
\end{remark}

Now, for the measure $\mu_0$ we check the extremality condition: $3\kappa\gamma<1$. Let $p_{0}=\frac12$. For $0<\theta<1$ and $\frac{\theta}{z+\theta}\leq\frac12$, according to Proposition 1, this condition becomes
\begin{equation}\nonumber
3\kappa\gamma-1 = 3\cdot\frac{z}{z+\theta}\cdot\frac{z p_{0}}{(z-\theta)p_{0}+\theta} - 1 < 0,
\end{equation}
and for $\theta\geq 1$ and $\frac{\theta}{z+\theta}\geq\frac12$,
\begin{equation}\nonumber
3\kappa\gamma-1 = 3\cdot\frac{\theta}{z+\theta}\cdot\frac{\theta(1-p_{0})}{(z-\theta)p_{0}+\theta} - 1 < 0.
\end{equation}
By computer analysis, these inequalities hold for $\theta_{cr}^{(1)}<\theta<\theta_{cr}^{(2)}$ (see Fig. 2).

Thus, we have the following

\begin{theorem}\label{thm6} Let $k=3$. Then, for the HC-Blume-Capel model on a \textquotedblleft wand\textquotedblright\ graph, the measure $\mu_0$ is extremal for $\theta_c^{(1)}<\theta<\theta_c^{(2)}$, where $\theta_c^{(1)}\approx0.83$ and $\theta_c^{(2)}\approx1.226$. \end{theorem}

For the case $k=2$, the following result was established in \cite{27}:

\begin{remark}\label{rk4} Let $k=2$. Then, for the HC-Blume-Capel model on a \textquotedblleft wand\textquotedblright\ graph, the measure $\mu_0$ is extremal for
\begin{equation}\nonumber
\frac{1}{2}\sqrt[3]{4\sqrt{2}-4}<\theta<\frac{1}{2}\sqrt[3]{28+20\sqrt{2}}.
\end{equation} \end{remark}

\begin{remark}\label{rk5} From Remark \ref{rk4}, it does not necessarily follow that extremality conditions hold for the measures $\mu_0$ when $k\geq4$.\end{remark}

For the measures $\mu_1$ and $\mu_2$, only the following result is known from \cite{27}:

\begin{remark}\label{rk6} Let $k=2$. Then, for the HC-Blume-Capel model on a \textquotedblleft wand\textquotedblright\ graph, the measures $\mu_1$ and $\mu_2$ are extremal for
\begin{equation}\nonumber
\sqrt[3]{\sqrt{2}-1+\frac{\sqrt{2\sqrt{2}-2}}{2}}<\theta<1,
\end{equation}
where $\sqrt[3]{\sqrt{2}-1+\frac{\sqrt{2\sqrt{2}-2}}{2}}\approx0.95$. \end{remark}

\begin{remark}\label{rk7} For the measures $\mu_1$ and $\mu_2$ with $k\geq3$, the question of their (non)extremality remains open.\end{remark}

\section{Conclusion}\label{conclusion}
In this paper, we investigated TISGMs for the HC-Blume-Capel model on a wand-type graph embedded in a Cayley tree of arbitrary order \( k \geq 2 \). It is known that there is the exact critical value \( \theta_{cr} \), which determines the number of TISGMs in the system. Specifically, for \( \theta \geq \theta_{cr} \), a unique TISGM exists, while for \( \theta < \theta_{cr} \), exactly three TISGMs emerge. Furthermore, we addressed the problem of extremality and non-extremality of these measures using the Kesten-Stigum criterion and other analytical methods.

Our results extend previous findings for special cases (such as \( k = 2 \)) to general values of \( k \), providing a more comprehensive understanding of the phase structure and critical behavior of the model. The identification of extremal and non-extremal Gibbs measures offers insights into the underlying stochastic dynamics and phase transitions in the system.

Future research directions include the study of periodic Gibbs measures beyond the translation-invariant case, as well as the exploration of the impact of external fields and additional interactions on the extremality properties of the measures. Moreover, numerical simulations and further analytical approaches could help refine the exact conditions under which non-trivial phase transitions occur.

\section{Data Availability Statement}

Not applicable

\section{Conflicts of Interest} The authors declare that they have no conflict of interest

\section{Funding}
No funding

\end{document}